\documentclass[letterpaper, 10 pt, conference]{Attachment/ieeeconf}

\IEEEoverridecommandlockouts                              
                                                          
\usepackage[utf8]{inputenc}
\usepackage{amsmath}
\usepackage{mathtools}
\usepackage{amsfonts}
\usepackage{xcolor}
\usepackage[english]{babel}
\usepackage{stackengine}
\usepackage{amssymb}
\usepackage{mathrsfs}
\usepackage{algorithm, float}
\usepackage{algorithmic}
\usepackage{bbm}
\usepackage{booktabs} 
\usepackage{caption}
\usepackage{subcaption}

\newtheorem{lemma}{Lemma}

\newtheorem{definition}{Definition}
\newtheorem{Pro}{Problem}

\DeclareMathOperator*{\argmin}{arg\,min}
\newcommand{\vect}[1]{\mathbf{#1}}
\DeclareMathOperator{\E}{\mathbb{E}}

\title{Control Barrier Function Augmentation in Sampling-based Control Algorithm for Sample Efficiency}

\author{Chuyuan Tao$^{\dagger}$, Hunmin Kim$^{\dagger}$, Hyungjin Yoon$^{*}$, Naira Hovakimyan$^{\dagger}$, and Petros Voulgaris$^{*}$
\thanks{This work has been supported by the National Science Foundation (CNS-1932529) and UIUC STII-21-06.}
\thanks{$^{\dagger}$Chuyuan Tao, Hunmin Kim, and Naira Hovakimyan are with the Department of Mechanical Science and Engineering, University of Illinois at Urbana-Champaign, USA.
{\tt\small  \{chuyuan2, hunmin, nhovakim\}@illinois.edu}}%
\thanks{$^{*}$Hyungjin Yoon and Petros Voulgaris are with the Department of Mechanical Engineering, University of Nevada, Reno, USA.
{\tt\small  \{hyungjiny, pvoulgaris\}@unr.edu }}
}

\begin{document}

\maketitle

\begin{abstract}
    For a nonlinear stochastic path planning problem, sampling-based algorithms generate thousands of random sample trajectories to find the optimal path while guaranteeing safety by Lagrangian penalty methods. However, the sampling-based algorithm can perform poorly in obstacle-rich environments because most samples might violate safety constraints, invalidating the corresponding samples. To improve the sample efficiency of sampling-based algorithms in cluttered environments, we propose an algorithm based on  model predictive path integral control  and  control barrier functions. The proposed algorithm needs fewer samples and time-steps and has a better performance in cluttered environments compared to the original model predictive path integral control algorithm. 
    
    
    
    

\end{abstract}

\section{Introduction}
Path planning problem aims to find an efficient and collision-free path for a robot from a start position to a target position. There have been proposed various deterministic model-based approaches as solutions to this problem, such as model predictive control (MPC) \cite{paden2016survey}, \cite{falcone2007predictive}, and Bernstein polynomials \cite{choe2015trajectory}. However, these methods usually require calculating the derivative of the dynamics and the cost functions, which makes them computationally expensive when trying to solve nonlinear stochastic optimization problems. Some cost functions might not be differentiable over the entire state space, making them even more complicated. 

Sampling-based algorithms use sample trajectories simulated from the dynamical models. The types of sampling-based algorithms include: rapidly exploring trees (RRT)~\cite{lavalle2001randomized} constructs the tree incrementally from samples drawn from the configuration space; probabilistic roadmap~\cite{kavraki1996probabilistic, kavraki1998analysis} samples configuration in free configuration space and uses a local planner to connect them;  Markov chain Monte Carlo algorithm (MCMC) algorithms~\cite{gelfand1990sampling} are used to approximate the posterior distribution of a parameter of interest by random sampling in a probabilistic space; and model predictive path integral control (MPPI) algorithm~\cite{williams2016aggressive}  uses random diffusion trajectories to determine the stochastic optimal trajectory. These methods were shown to be effective in practical robotic applications~\cite{yang2014spline, lavalle2006planning, williams2018information}, where parallel computing is applicable to sample thousands of trajectories fast in parallel.

However, the uniformly randomized sampling algorithms can become sample-inefficient in extreme circumstances~\cite{wang2018learning}, such as in a small region connecting two large free space areas, i.e., narrow passage problems. It takes more time for sampling-based algorithms to identify a safe path passing through the narrow passage, though the portion of the narrow passage is small compared to the whole configuration space. To address this issue, RRT was developed to utilize the previous exploration paths and generate incremental random paths starting from the previous terminal position~\cite{li2002incremental}. Alternatively, an adaptive sampling method that adaptively adjusts the sample region to avoid collisions was proposed in~\cite{jaillet2005adaptive} as an extension of RRT. Another improved version of RRT that restricts the permitted sampling region based on the previous exploring results was proposed in \cite{yershova2005dynamic}.

In this paper, we aim to develop a MPPI-CBF algorithm that restricts the sampling distribution in the MPPI algorithm \cite{williams2018information} by using the control barrier function (CBF) to improve the sample efficiency of sampling-based algorithms for narrow passage problems. We trade off the computation time for solving a CBF-based optimization problem at every time step to guarantee sampling from a safe region. Therefore, the proposed MPPI-CBF algorithm needs fewer samples and time-steps to achieve similar performance to the MPPI algorithm.

\subsection{Related Work}
Stochastic model predictive control (SMPC) is an optimal control strategy for nonlinear stochastic systems subject to chance constraints. As one of the SMPC strategies, the MPPI algorithm provides an effective solution for finite horizon nonlinear stochastic optimization problems through parallel computing, which allows real-time control of a wide class of robotic applications \cite{williams2018information}, \cite{williams2017information}. The MPPI algorithm has been widely used in nonlinear path planning problems. Reference \cite{pravitra2020} uses robust control to help the MPPI algorithm be more robust when the simulated systems are different from the true dynamics. In \cite{williams2017information}, the author uses machine learning algorithms to handle more general nonlinear control systems. Most existing MPPI algorithms use negative rewards as soft constraints to avoid collision.

Another way to guarantee collision avoidance in path planning problems is to utilize CBF constraints \cite{ames2019control}, which is widely used in safe control problems. CBF methods can also be applied to systems with uncertainties. In \cite{cheng2019end}, the author uses CBF to guide the reinforcement learning (RL) algorithm to explore in a safe policy set. In \cite{zhao2020adaptive}, the paper uses a robust control method in CBF reformulate a method that guarantees the robust safety of the system. Furthermore, in \cite{song2021generalization}, a framework generating optimal control algorithm while guaranteeing safety on the multi-agent system is discussed. 


\section{Notations}
We use the subscript $t$ to denote the discrete time index and the superscript $i$ to denote the sample index, e.g., $x^i_t$ denotes the state of the $i$\textsuperscript{th} sample at time $t$. Let $\tau : [ t_0, T] \rightarrow \mathbb{R}^n$  represent a sample trajectory of the system. We use $\tau_i$ to represent the sample trajectory when the sample index is $i$. For a symmetric matrix $\mathbf{A}$, $\mathbf{A} \succ 0$ and $\mathbf{A}\succeq 0$ indicate that $\mathbf{A}$ is positive definite and positive semi-definite, respectively. The matrix $\mathbf{I}$ denotes the identity matrix with an appropriate dimension. We use $\Vert \cdot \Vert$ to denote the standard Euclidean norm for vector or an induced matrix norm if it is not specified, $\E[\cdot]$ to denote the expectation, and $\text{Var}(\cdot)$ to denote the variance. We use $L$ to represent the Lie derivative of a function. We define $V=(v_0,v_1,\dots,v_{T-1})$ as a sequence of actual control inputs over some number of timesteps $T$, and $U=(u_0,u_1,\dots,u_{T-1})$ as a sequence of mean input variables. The list of the notations is summarized in Table~\ref{tab:notation}.

\begin{table}[!t] \vspace{0.15cm}
\caption{List of Notations}
\label{tab:notation}
\begin{center}
\begin{small}
\begin{tabular}{cl}
\toprule
Notation & Description \\
\midrule


$t$    & Time index for path-planning.\\ 
$i$    & Sample index for path-planning.\\
$\vect{x}_{t}^i$  & States of the agent for $i_{th}$ sample at time $t$.\\ 
$\tau_i$ & Sample trajectory when the index is $i$. \\
$V$ & Real control input sequence.\\
$T$ & Sample trajectories time horizon.\\
$N$ & Number of the sample trajectories.\\
$\omega$ & The weights of each sample.\\

\bottomrule
\end{tabular}
\end{small}
\end{center}
\vskip -0.1in
\end{table}


\section{Problem Statement} 

Consider the stochastic dynamic system to be in the following form:
\begin{equation}
    dx = f(x_t,t) dt + G(x_t, t)(dW(t) +u(x_t,t)dt), \label{eq:dynamics}
\end{equation}
where $x_t \in \mathbb{R}^n$ represents the state of an agent at time $t$, and let $u(x_t,t)\in \mathbb{R}^m$  be the control input of the system, and $dW(t) \in \mathbb{R}^p$ is a Brownian motion disturbance on the dynamical system.

We suppose that the finite time horizon optimization problem has a quadratic control cost and an arbitrary state-dependent cost. The value function $V(x_t,t)$ is then defined as
\begin{equation}\label{eq:stochastic_prob}
    \min_u \mathbb{E}_{\mathbb{Q}}\left[ \phi(x_T) + \int_t^T(q(x_t,t)+\frac{1}{2}u^TR(x_t,t)u)dt \right],
\end{equation}
where $R(x_t,t)$ is a positive definite matrix, $\phi(x_t)$ denotes a terminal cost, and $q(x_t,t)$ is a state-dependent cost. We assume that the operating environment is obstacle rich and let $C$ represent a known safe set, i.e. if $x_t \in C$ for $\forall t\geq 0$, then the system is safe.

\begin{Pro}\label{pro1}
Given $x_0 \in C$, the problem is to develop a sampling-based algorithm for the optimization problem in (\ref{eq:stochastic_prob}) subject to \eqref{eq:dynamics}. while guaranteeing $x_t \in C$ for $\forall t\geq0$.

\end{Pro}

\section{Control Barrier Function based Model Predictive Path Integral Control (MPPI-CBF)}\label{sec:algorithm}


To address Problem~\ref{pro1} for cluttered environments, we develop a MPPI-CBF algorithm where CBF is applied to determine the sampling distribution of each step in the MPPI algorithm to guarantee safety constraints in sampling based methods. Since the proposed algorithm is based on MPPI, we first introduce a summary of the MPPI algorithm in Section \ref{sec:MPPI}. Section \ref{sec:MPPI-CBF} discusses the naive combination of the MPPI algorithm and CBF. Lastly, Section \ref{sec:MPPI-CBF_new} proposes the MPPI-CBF algorithm, which achieves an efficient sampling by trading off it with computational complexity.

\subsection{Model Predictive Path Integral Control (MPPI)}\label{sec:MPPI}

MPPI algorithm \cite{williams2018information}, \cite{williams2017information} is a sampling-based algorithm and provides a framework to solve nonlinear control problems efficiently for non-cluttered environments. By parallelizing the sampling step using Graphics Process Unit (GPU), the MPPI algorithm could sample thousands of trajectories of nonlinear complex dynamical systems in milliseconds~\cite{williams2016aggressive}. The algorithm does not require finding derivatives of the dynamical system. Instead, sampled control inputs are propagated through the nonlinear dynamical system, and this allows to evaluate the cost for the corresponding sample trajectory.

In~\cite{williams2017information}, the MPPI algorithm applies Jensen's inequality to the free energy function of the control system to reach a tight bound. Next, by the tight bound, we define the density function of the approximate optimal distribution $\mathbb{Q}^*$ to be:
\begin{equation}
    q^*(V)=\frac{1}{\eta} \exp{(-\frac{1}{\lambda}S(\tau))}p(\tau),
\end{equation}
where $p(\tau)$ is the Gaussian distribution density function of $\mathbb{P}$, $\eta$ is a normalizing constant, and the function $S(\tau)= \phi(x_T)+\int_{t_0}^Tq(x_t,t)dt$ is the cost-to-go of the state-dependent cost.

Following the approach in~\cite{williams2016aggressive}, we have a new optimization problem by pushing the controlled distribution as close as possible to the optimal control input:
\begin{equation} \label{eq:kl-equation}
    u^*(\cdot)=\argmin_{u(\cdot)}\mathbb{D}_{KL}(\mathbb{Q}^* || \mathbb{Q}),
\end{equation}
where $\mathbb{D}_{KL}(\cdot | \cdot)$ denotes KL-divergence, which represents the distance between two distributions. The KL-divergence can be simplified to a convex quadratic function of $u$ by taking the density function $p(V)$ and $q(V)$ into the equations. By finding the gradient of convex equation and setting it equal to zero, we have the optimal control input as:
\begin{equation}
    u^*_t=\int q^*(V)v_t dV.
\end{equation}
Using the importance sampling weights and Monte-Carlo estimate \cite{williams2018information}, the control update law is given as:
\begin{equation}\label{eq:control_update}
    u^*(x_t,t) = u(x_t, t) +\frac{\mathbb{E}_q \left [\exp(-\frac{1}{\lambda}S(\tau))\delta u  \right ] }{\mathbb{E}_q \left [ \exp(-\frac{1}{\lambda}S(\tau)) \right ]},
\end{equation}
where $\lambda$ is the temperature parameter of Gibbs distribution, $\delta u$ is the random control input, and $\tau$ is the sample trajectory. The above equation \eqref{eq:control_update} can be approximated as:
\begin{equation} \label{eq:discrete_update_law}
    \hat{u}^*(x_{t_i},t_i) \approx u(x_{t_i},t_i)+\frac{\sum^K_{k=1}\exp(-\frac{1}{\lambda}S(\tau_{i,k}))\delta u_{i,k}}{\sum^K_{k=1}\exp(-\frac{1}{\lambda}S(\tau_{i,k}))}.
\end{equation}

It is common to use an indicator function as a penalty in the cost-to-go function $S(\tau)$ in solving path planning problems. However, in narrow passage problems, the majority of the random sample trajectories will violate the safe constraints, and have a large negative penalty. Thus, the sample efficiency will become low and the algorithm needs more samples to finish the narrow path problems. We use the control barrier function instead of adding penalties in the cost function to guarantee safety and increase sample efficiency.

\subsection{CBF-Based Compensating Control with MPPI} \label{sec:MPPI-CBF}
We first introduce a straightforward shielded MPPI-CBF algorithm to illustrate the necessity of our framework. The concept is akin to shielded RL \cite{fisac2018general}, which uses the CBF to compensate for the unsafe actions of the controller.  

The CBF algorithm uses a Lyapunov-like argument to provide a sufficient condition for ensuring forward invariance of the safe set $C$. The following defines the CBF function.
\begin{definition} \cite{ames2019control} \label{def:CBF}
Let $C \in \mathbb{R}^n$ be the super-level set of a continuously differentiable function $h \in \mathbb{R}^n \rightarrow \mathbb{R} $. Then $h$ is a control barrier function if there exists an extended class $K_\infty$ function $\alpha$ such that for the control system:
\begin{equation}
    \sup_{u\in U} [L_f h(x)+ L_gh(x)u]\geq -\alpha(h(x)),
\end{equation}
for all $x\in D$.
\end{definition}

We assume a feedback controller is defined as $u= k(x)$ and formulate a quadratic programming problem (QP) based on the CBF constraints to guarantee safety:
\begin{equation}\label{eq:qp}
    \begin{aligned}
u_{\text{CBF}}=\argmin_u\quad & \frac{1}{2} \left \| u - k(x) \right \|^2_2 \\
\textrm{s.t.} \quad & L_f h(x) + L_gh(x)u \geq -\alpha(h(x)).\\
\end{aligned}
\end{equation}
The shielded MPPI-CBF algorithm uses the CBF as a safe filter for the control output of the MPPI controller, and it is shown in Fig. \ref{fig:naive_idea_structure}. The control update law becomes
\begin{equation}
    u_k(x_t) = u_{\text{MPPI}}(x_t)+u_{\text{CBF}}(x_t, u_{\text{MPPI}}(x_t)).
\end{equation}

The MPPI algorithm provides the nominal control input and using the CBF-based optimization to find a minimum control intervention renders the safe set forward invariant.
\begin{figure}
    \centering
    \includegraphics[width=8cm]{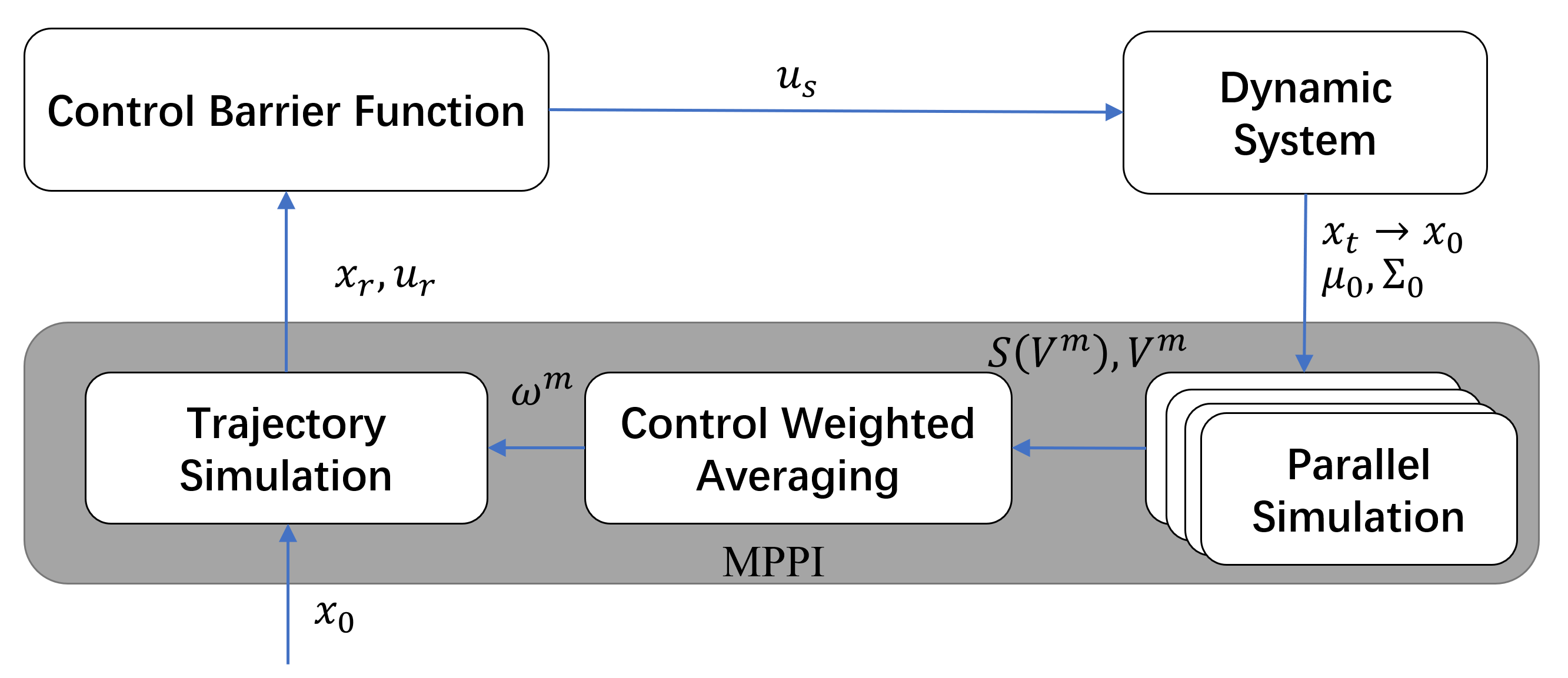}
    \caption{Control architecture combining MPPI control with CBF to compensate for unsafe actions without considering the influence of CBF on the weights of the samples. The method guarantees safety but does not yield identical  performance level to MPPI.}
    \label{fig:naive_idea_structure}
\end{figure}
Note that comparing to the CBF optimization in \eqref{eq:qp}, the constraint of the stochastic system is not deterministic anymore. Therefore, we calculate the trust region of the control input by a chance constraint to ensure the safety of the robot. Therefore, the stochastic quadratic programming at each time step is:
\begin{equation} \label{eq:simple_QP}
    \begin{aligned}
    &\argmin_{u} \left \|u - u_{\text{MPPI}} \right \|^2_2 \\
    \textrm{s.t.}\quad  &Pr(L_gh(x)u \geq -\alpha(h(x)) - L_f h(x)) \geq 1-\delta. \quad
    \end{aligned}
\end{equation}

A critical issue for this shielded MPPI-CBF algorithm is that the CBF constraints override the control inputs of the MPPI algorithm without considering the sampling distribution. Thus, in the presence of the safety challenges, the CBF constraints hinder the MPPI algorithm from exploring a feasible region. As a result, the CBF-based compensation controller with the MPPI algorithm prevents the robot from completing the mission. Therefore, we develop an MPPI-CBF algorithm based on sampling in safe regions instead of sampling uniformly and overriding the control input.

\subsection{CBF-Based Trust Region Sample Control with MPPI} \label{sec:MPPI-CBF_new}
Instead of using CBF as a safe filter for the control input, we propose a new MPPI-CBF algorithm using CBF to determine the safe regions to sample. First, we define the initial control input samples from a Gaussian distribution $u_0 \sim \mathcal{N}_k( \mu_0, \Sigma_0)$. Next, we assume the safe control input regions can be approximated by a Gaussian distribution $u_t^i \sim \mathcal{N}_k (\mu_t^i, \Sigma_t^i)$. Note the safe control input region is a time-varying and state-dependent set.


Also, comparing to the QP optimization problem in \eqref{eq:simple_QP}, no nominal control is provided. Hence, we define a new convex cost function for the mean and variance of the sampling distributions:
\begin{equation}
    \mathscr{L}(\mu, \Sigma) = \left \| \mu - \mu_0 \right \|_1 + \left \| \Sigma -\Sigma_0 \right \|_p,
\end{equation}
where $\mu, \Sigma$ are the variables to optimize and $\left \| \cdot \right \|_p$ represents any matrix norm. We will use the above cost function for the CBF-based optimization to calculate the smallest change to the original control input distribution. We will find the safe mean and variance for the new sample trajectories at each time step. This structure is illustrated in Fig. \ref{fig:final_idea_structure}.
\begin{figure}
    \centering
    \includegraphics[width=8.5cm]{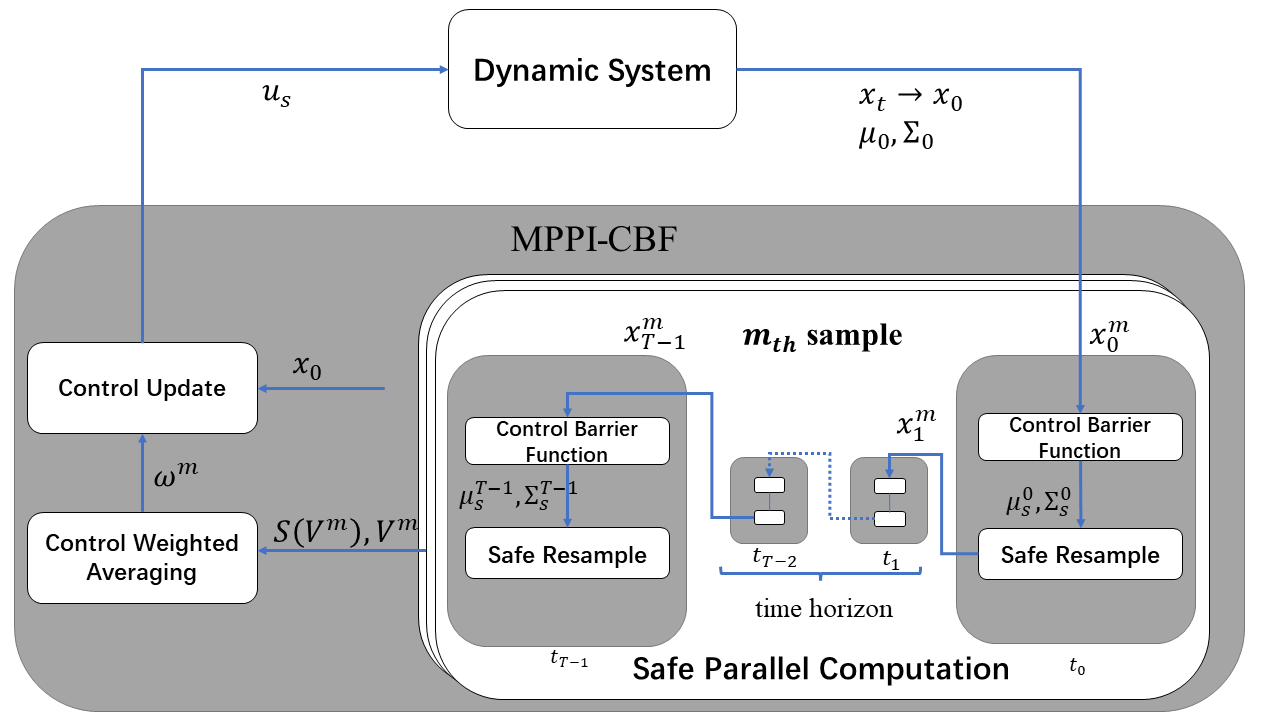}
    \caption{We solve the SDP optimization~\eqref{eq:SDP for MPPI_CBF} for each sample and at each time step to find a safe mean and variance and sample a new safe control input at this time step. Next, we propagate through the dynamics of the robot and repeat at the next time step. Therefore, we can sample safe trajectories and calculate their weights for the control update law. Finally, we  update the control policy based on the weights.}
    \label{fig:final_idea_structure}
\end{figure}

Combining with the probabilistic CBF constraints defined in \eqref{eq:simple_QP}, we reformulate the optimization problem as
\begin{equation}\label{eq_optimal_original}
    \begin{aligned}
    &\argmin_{\mu, \Sigma} \left \| \mu - \mu_0 \right \|_1 + \left \| \Sigma -\Sigma_0 \right \|_p  \\
    \textrm{s.t.}\quad  &\text{Pr}(L_gh(x_t^i)u_t^i \geq -\alpha(h(x_t^i)) - L_f h(x_t^i)) \geq 1-\delta .
    \end{aligned}
\end{equation}

Let $A_{i,t} = L_g h(x_i^t)$, $b_{i,t}=-\alpha(h(x_i^t)) - L_f h(x_i^t)$. Then the left hand side of the control barrier function constraint can be considered as a new random variable $A_{i,t} u_t^i$, having mean $A_{i,t} \mu_t^i$ and variance $A_{i,t}\Sigma_t^i A_{i,t}^T$.
Using the chance constraint with the trust region probability $1-\delta$, the control barrier function constraint is simplified to
\begin{equation}
    A_{i,t} \mu - c A_{i,t} \Sigma A_{i,t}^T \geq b_{i,t},
\end{equation}
where $c$ is the trust region corresponding to the probability $1-\delta$. To apply the CBF on the mean and variance, we also need to consider another constraint that the variance matrix of the control input random variable $u_t^i$ must be positive semidefinite
\begin{equation}
    \Sigma \succeq 0.
\end{equation}
Then optimization problem (\ref{eq_optimal_original}) is reformulated to the following convex optimization problem:
\begin{equation}\label{eq:convex_opt}
    \begin{aligned}
        \argmin_{\mu, \Sigma} &\left \| \mu - \mu_0 \right \|_1 + \left \| \Sigma -\Sigma_0 \right \|_p \\
        \textrm{s.t.}\quad & A_{i,t} \mu -c A_{i,t} \Sigma A_{i,t}^T \geq b_i \\
        & \Sigma \succeq 0.
    \end{aligned} 
\end{equation}
This convex problem can be further transformed into a semidefinite programming problem (SDP), which can be solved more efficiently \cite{peng2012faster}.

A standard SDP problem is written as
\begin{equation}
\begin{aligned}
    \min \quad &C\cdot X \\
    \text{s.t.}\quad  &A_i \cdot X = b_i \quad \text{for } i = 1, \dots, n  \\
    &X \succeq 0,
\end{aligned}
\end{equation}
where $X\in S^n$, with $S^n$ being the space of real symmetric $n\times n$ matrices. The vector $b\in \mathbb{R}^m$, and the matrices $A_i\in S^n$ and $C\in S^n$ are given parameters. While the standard SDP problem can also be cast in various ways, one of the most common formats is the inequality form is given in \cite{wolkowicz2012handbook}:
\begin{equation}\label{eq:sdp_ineq}
    \begin{aligned}
        \min \quad & c^T  x \\
        \text{s.t.}\quad  & \sum_{i=1}^m x_i A_i \preceq B ,
    \end{aligned}
\end{equation}
where c is a vector and $A_i, B \in S^n$ are matrices. The problem is easy to transform to SDP problem by introducing a slack variable $S$ to change the inequality to equality. Then the problem is written as:
\begin{equation}
    \begin{aligned}
        \min \quad & c^T  x \\
        \text{s.t.}\quad  & \sum_{i=1}^m x_i A_i +S = B \\
        &S \succeq 0.
    \end{aligned}
\end{equation}

\begin{lemma}
The convex optimization problem \eqref{eq:convex_opt} can be cast as a SDP problem.
\end{lemma}

\begin{proof}
The first part of the cost function in~\eqref{eq:convex_opt} is a vector norm for the mean difference. We rewrite it as $\sum_i |\mu_i-\mu_0|$, and redefine the variable $s_i = |\mu_i -\mu_{0i}|$. Then the cost function becomes a linear combination of the variables $s_i$. Using $A(x) = A_0 + x_1 A_1 + \cdots + x_n A_n$ to represent the symmetric matrix, where $A_i \in \mathbb{R}^{m\times m}, i = 0, \dots, n$, the second part of the cost function becomes $\min\|A(x)\|_p$, This can be further reformed as follows with a new variable $t>0$:
\begin{equation}\label{eq:lem}
    \begin{aligned}
        \min \quad &t \\
        \text{s.t} \quad & \begin{bmatrix}
        tI & A(x)\\
        A(x)^T & tI
        \end{bmatrix}\succeq 0.
    \end{aligned}
\end{equation}
Then, the cost functions of the first and second parts are linear and the constraints in~\eqref{eq:convex_opt} and those in~\eqref{eq:lem} follow the form of SDP in~\eqref{eq:sdp_ineq}. This completes the proof.
\end{proof}

For any $\Sigma_0$, there exists matrix $P_0$ such that $\Sigma_0 = P_0 P_0^T$. Furthermore, for any matrix $P$, $\Sigma = PP^T$ becomes a symmetric definite matrix. Using this fact, we can simplify the constraint on the variance matrix, and the problem~\eqref{eq:convex_opt} becomes
\begin{equation}
        \begin{aligned}
        \argmin_{\mu, P} &\left \| \mu - \mu_0 \right \|_1 + \left \| P -P_0 \right \|_p  \\
        \textrm{s.t.} & \begin{pmatrix}
                    I & \sqrt{c}A_{i,t}P\\ 
                    \sqrt{c}P^TA_{i,t}^T & \mu A_{i,t}-b
                    \end{pmatrix} \succeq 0.
        \end{aligned} \label{eq:SDP for MPPI_CBF}
\end{equation}

The solution of optimization (\ref{eq:SDP for MPPI_CBF}) provides the safe mean $\mu_s$ and the safe variance $\Sigma_s= PP^T$ of the sample regions. Next, we will sample a new control input based on the new distribution and propagate through the dynamics \eqref{eq:dynamics}  to get the state $\tau_{i,k}$ of the new safe trajectories $\tau_{k}$. It guarantees the sample trajectories are in the safe set $C$ defined by the control barrier function with probability $1-\delta$ for each time step. We will calculate the cost of each sample trajectory by using function $S(\cdot)$ and using the following equation to calculate the weight of each trajectory:
\begin{equation}\label{eq:weight_calculate}
    \omega_{i,k} = \exp (-\frac{1}{\lambda}S(\tau_{i,k})).
\end{equation}
We define the sample weight $\omega_{i,k} = \exp (-\frac{1}{\lambda}S(\tau_{i,k}))$. We use the control update law:
\begin{equation}\label{eq:update_law}
    u^*(x_{t_i},t_i) \approx u(x_{t_i},t_i)+\frac{\sum^K_{k=1}\omega_{i,k}\delta u_{i,k}}{\sum^K_{k=1}\omega_{i,k}}.
\end{equation}
The algorithm is summarized in Algorithm~\ref{algo1}.

Comparing to the original MPPI algorithm, the MPPI-CBF algorithm needs more computation time per sample because it requires solving a SDP optimization problem every time step. However, using parallel computing, we can solve SDP optimization problems efficiently. In \cite{warmuth2008randomized} \cite{helmberg1996interior} \cite{peng2012faster}, it is shown that using parallel computing, the SDP problems can be solved almost as efficiently as linear programming.

\begin{algorithm}
\caption{MPPI-CBF algorithm}
\begin{algorithmic}\label{algo1}
\WHILE{task is not completed}
    \FOR{$k \leftarrow 0$ to $K-1$}
        \STATE{$x \gets x_0$;}
        \FOR{$t \gets 1$ to $T$}
            \STATE{Solving SDP in \eqref{eq:SDP for MPPI_CBF} to get $\mu_s, \Sigma_s$;} 
            \STATE{Generate control variations $\delta u_{i,k}$ $\sim \mathscr{N}(\mu_s,\Sigma_s)$;}
            \STATE{Simulate discrete dynamic $x_t^k = x_{t-1}^k + (f+G(u_{i}+\delta u_{i,k}))\delta t$ to obtain $\tau_i^k$;}
            \STATE{Calculate cost function $S(\tau_i^k)$;}
        \ENDFOR
    \ENDFOR
    \STATE{$\beta \gets \min_k[S(x_t^k)]$;}
    \STATE{Get resample weights $\omega_{i,k}$ using \eqref{eq:weight_calculate};}
    \STATE{Update control input using $\omega_t$ and $\delta u_{i,k}$ using \eqref{eq:update_law};}
    \STATE{Send $u_{t_0}$ to actuator;} 
    \FOR{$i \gets 0$ to $T-2$}
        \STATE{$u_i = u_{i+1}$;}
    \ENDFOR
    \STATE{$u_{N-1}=u_{init}$;}
\ENDWHILE
\end{algorithmic}
\end{algorithm}

\section{Experiments} \label{sec:exper}
In this section, we implement the algorithm on a nonlinear control-affine system. We design two different path planning simulation scenarios, single-obstacle avoidance and multi-obstacle narrow passage problems. We validate the effectiveness of the proposed methodology in the aforementioned two scenarios, from the perspective of safety guarantees and algorithm performance, by comparing with the MPPI-based algorithm.
\subsection{Unicycle Dynamics}\label{sec:uni-dyn}
We consider a two-dimensional unicycle model in the simulation. In particular, the dynamical model of a unicycle is represented as:
\begin{equation}
    \begin{bmatrix}
\dot {x}\\ 
\dot {y}\\ 
\dot {\theta}
\end{bmatrix} = \begin{bmatrix}
\cos \theta & 0\\ 
\sin \theta & 0\\ 
0 & 1
\end{bmatrix}\begin{bmatrix}
v + \delta_v\\ 
\omega + \delta_\omega
\end{bmatrix},
\end{equation}
where $x, y$ are the coordinates, $\theta$ is the angle, $v$ is the linear velocity and $\omega$ is the angular velocity of the unicycle model. $\delta_v, \delta_\omega \sim \mathscr{N}(\mu_0, \Sigma_0)$ denote the random perturbation to the control input. The time step for the discrete time control update law \eqref{eq:discrete_update_law} is $\Delta t = 0.05 s$.

\subsection{Unicycle Single-Obstacle Avoidance Problem}
In the first simulation, the path planning problem starts from position $(0,0)$ and aims to reach the target position $p_{d}$ at $(x_{d},y_{d})$. The unicycle needs to avoid a circular obstacle centered at $(O_x, O_y)=(2.2,2.0)$, with the radius of $O_r=0.5$. This control objective is enforced as a safety-critical constraint through the following CBF:
\begin{equation}
    h(x,y) =  (x-O_x)^2 + (y-O_y)^2 - O_r ^2.
\end{equation}
The running cost function is designed in the form of:
\begin{equation}
    q(x) = 10 \left \| p_{\text{d}} - p \right \|_2^2 +  \left \|v_{\text{d}} - v\right \|_2^2,\label{eq:sim_cost_function}
\end{equation}
where $v_{\text{d}}$ is the desired velocity, and $p(x,y)$ is the current position of the unicycle. The initial mean and variance matrix of the noise distribution are: $\mu_0 = [0,0], \Sigma_0 =I$. The target is set at $(4,4)$ and the desired velocity is $v_d=2$. We first use 50, 100, 200, and 500 samples and 20 time steps in the simulation. The probability of the trust region is $1-\delta = 99.8 \%$. The result is shown in Fig. \ref{fig:single_CBF_traj}. As Fig. \ref{fig:single_CBF_traj} demonstrates, the MPPI-CBF algorithm can always drive the unicycle to the target position via an optimal path with various sample sizes. Meanwhile, when the samples are sufficient (i.e., greater than 100 samples in the simulated scenario), the proposed algorithm also guarantees safety by reserving a margin around the obstacle.

We repeat the same simulation with the MPPI algorithm for comparison. We use a similar cost function form but introduce an indicator function term for penalization whenever the trajectory violates the safety set; then the revised cost function is in the form of:
\begin{equation}
    q(x) = 10 \left \| p_{\text{d}} - p \right \|_2^2 +  \left \|v_{\text{d}} - v\right \|_2^2 +10000*\mathbf{1}_{p \in \overline{C}},
\end{equation}
where $\overline{C}$ is the complementary to the safe set $C$ over $\mathbb{R}^2$, $\mathbf{1}$ is the indicator function, and $C=\left \{(x,y)|(x-O_{x})^2 + (y-O_{y})^2 - O_{r} ^2 \geq 0 \right \}$.
The result of the MPPI algorithm is shown in Fig. \ref{fig:single_MPPI_traj}. In this single-obstacle avoidance scenario, it is observed that the MPPI algorithm can also guarantee safety by simply adding an indicator term in the cost function. However, comparing to the MPPI-CBF algorithm, the trajectories from the MPPI algorithm are more aggressive, which can be explained as the MPPI algorithm samples uniformly in the configuration space, whereas the MPPI-CBF algorithm samples in a safe but more conservative space. However, the MPPI-CBF algorithm has a more conservative result in this considered scenario. The two algorithms need almost the same total number of time steps to finish the path planning problem. We can see from Fig. \ref{fig:cost_single} where we plot the cost of each algorithm. Beyond the conservativeness of the running trajectory, there is no significant difference between the proposed algorithm and the baseline MPPI algorithm from the perspective of cost-reduction profile, which is illustrated in Fig. \ref{fig:cost_single}.
\begin{figure}[ht] 
\begin{subfigure}{.24\textwidth}
  \centering
  \includegraphics[width=1.\linewidth]{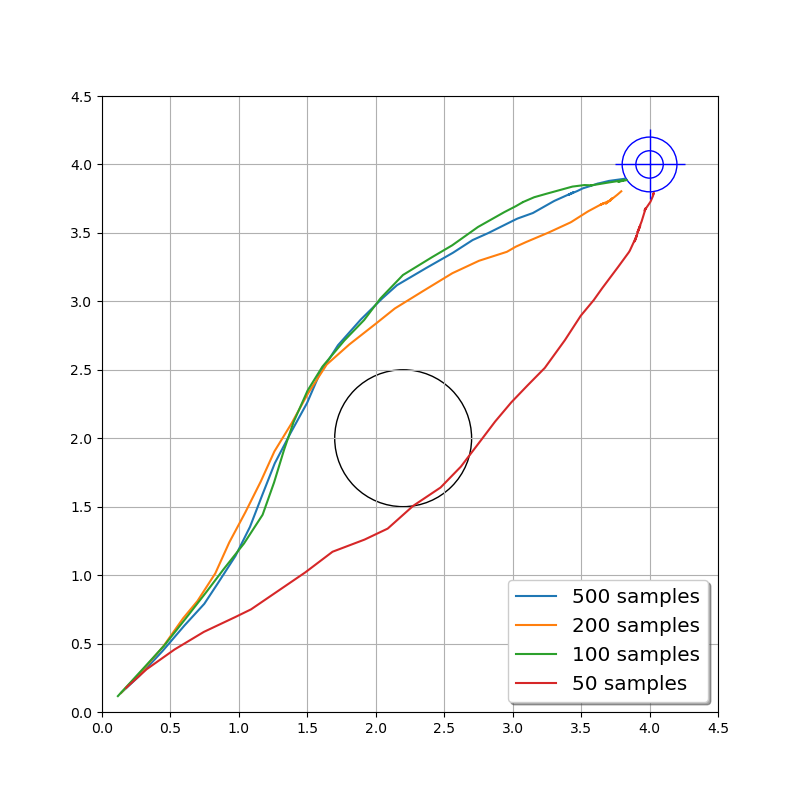}
  \caption{MPPI-CBF algorithm for single-obstacle simulation}
  \label{fig:single_CBF_traj}
\end{subfigure}
\begin{subfigure}{.24\textwidth}
  \centering
  \includegraphics[width=1.\linewidth]{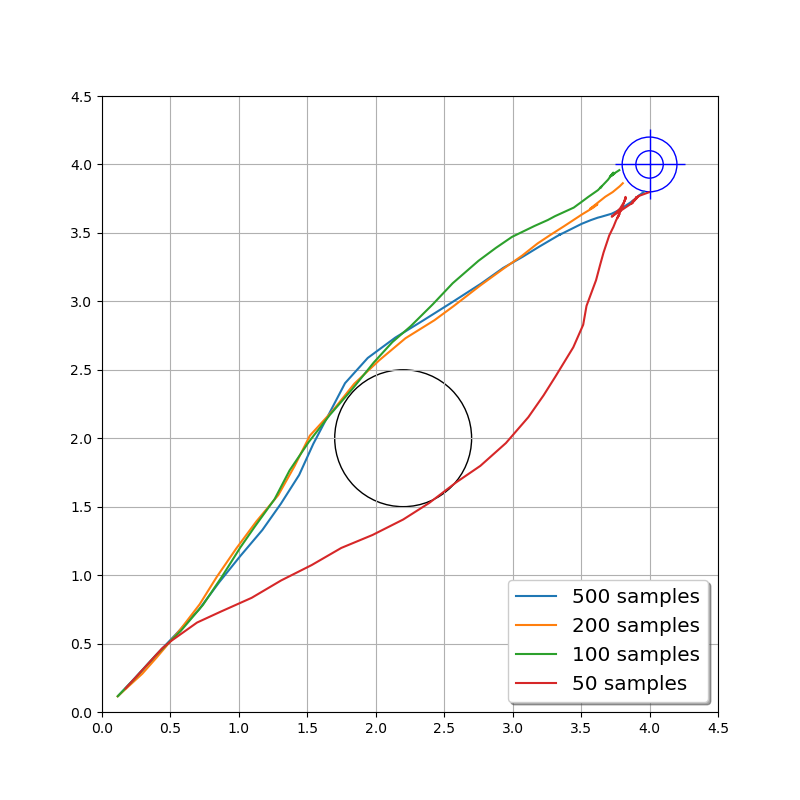}  
  \caption{MPPI algorithm for single-obstacle simulation}
  \label{fig:single_MPPI_traj}
\end{subfigure}
\caption{We use 500, 200, 100, and 50 samples to test in the single-obstacle avoidance scenario where the start point is the origin, the target is the blue circle at position $(4,4)$, and the obstacle is denoted by the black circle.}
\end{figure}

\begin{figure}[ht]
    \centering
    \includegraphics[width=6cm]{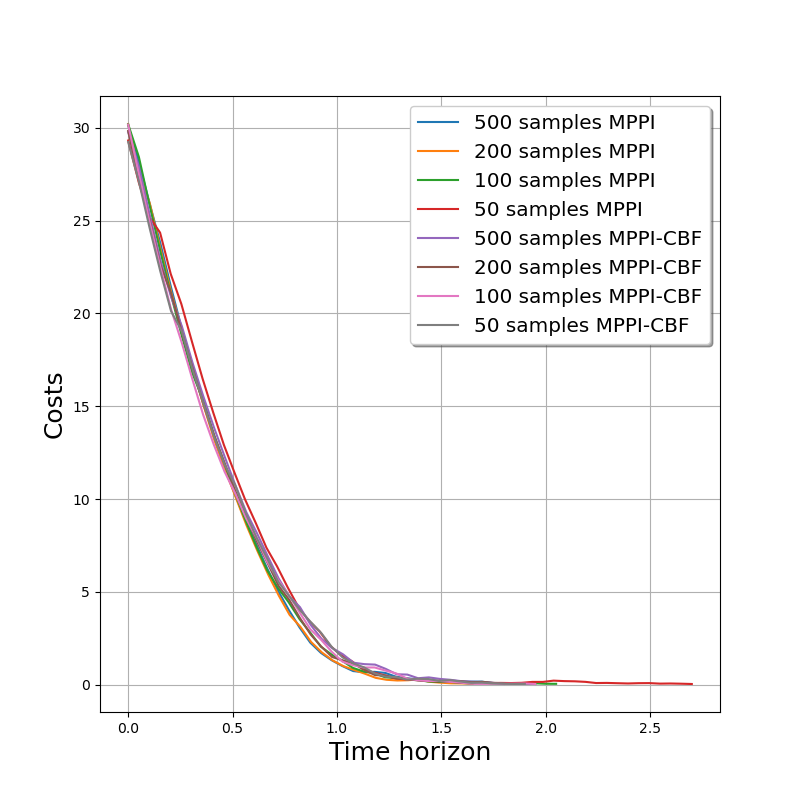}
    \caption{The cost of the MPPI algorithm and the MPPI-CBF algorithm. In the single-obstacle avoidance scenario, the MPPI and the MPPI-CBF algorithm spend almost the same time steps to solve the path planning problems.}
    \label{fig:cost_single}
\end{figure}

We also plot all the sample trajectories at time $t=0.35s$ in Fig. \ref{fig:single_plot}. The black circle represents the obstacle, and the sample trajectories are colored based on the costs. From the plots, we can conclude that the MPPI algorithm samples uniformly on the entire configuration space, and some sample trajectories violate the safety constraints. However, for MPPI-CBF, which chooses samples from a safe set with high probability, and from Fig. \ref{fig:single_plot}, it is observed that almost all the samples are in the safe set. Though conservative, the proposed MPPI-CBF algorithm provides more safe samples than the MPPI algorithm. This advantage will be more obviously shown in the narrow passage problem in the next section. 

\begin{figure}[ht] 
\begin{subfigure}{.24\textwidth}
  \centering
  \includegraphics[width=1.2\linewidth]{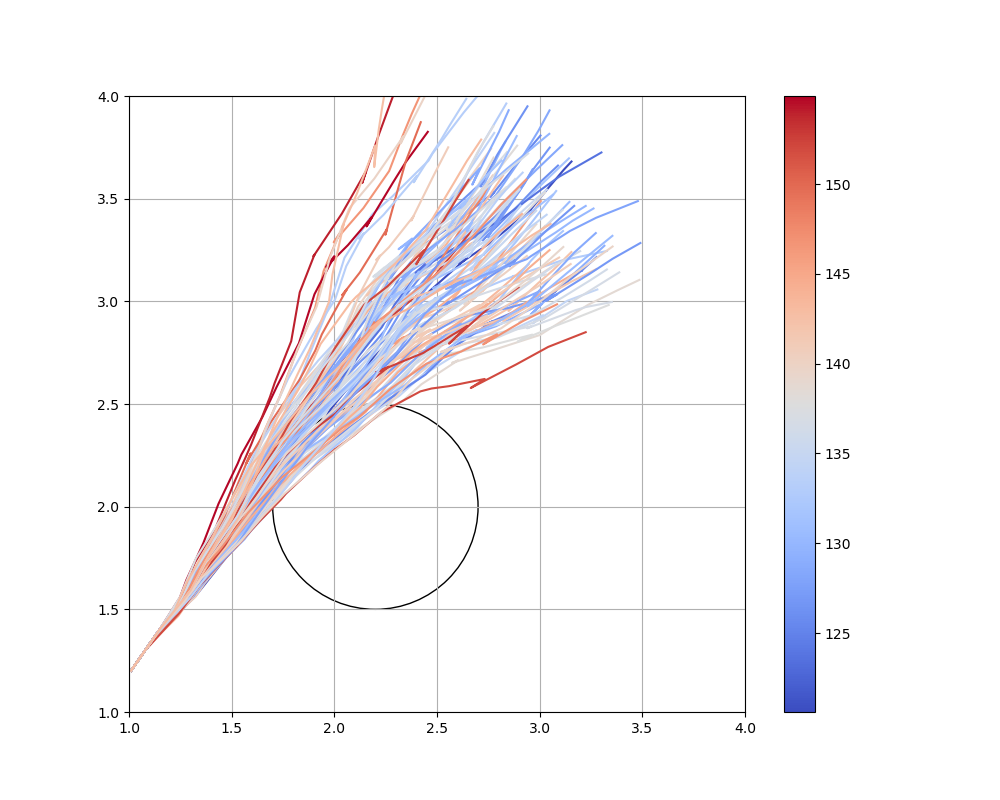}
  \caption{100 sample trajectories of MPPI algorithm at time $t=0.35s$}
  \label{fig:MPPI_sample}
\end{subfigure}
\begin{subfigure}{.24\textwidth}
  \centering
  \includegraphics[width=1.2\linewidth]{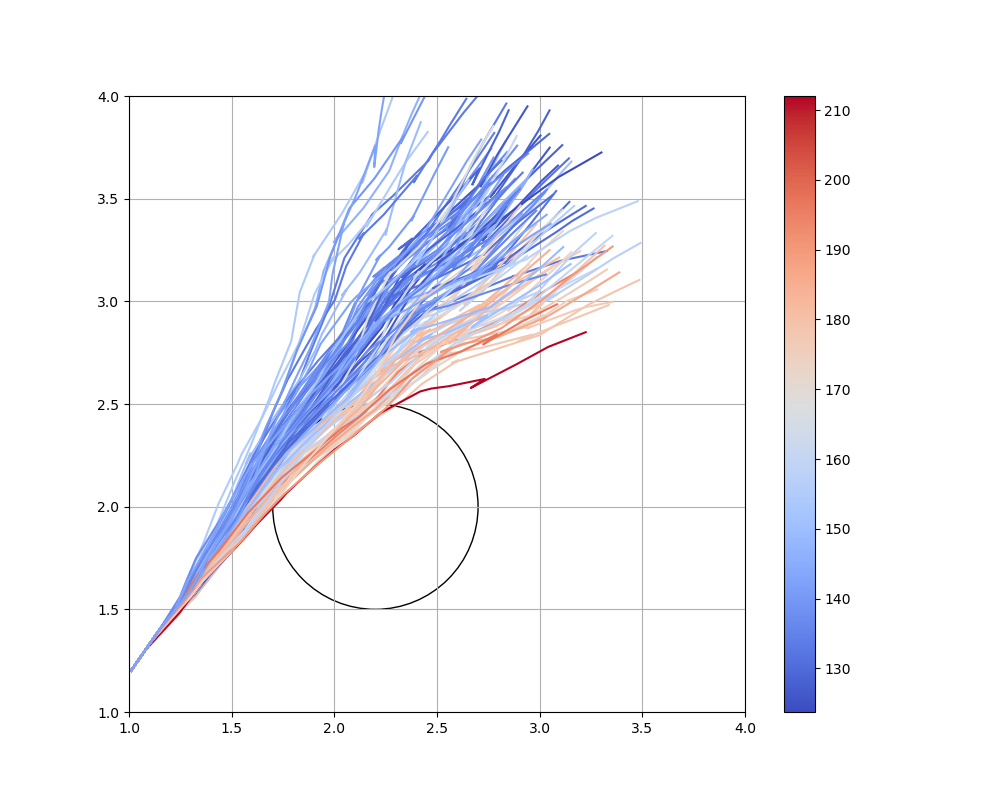}  
  \caption{100 sample trajectories of MPPI-CBF algorithm at time $t=0.35s$.}
  \label{fig:MPPI_CBF_sample}
\end{subfigure}
\caption{The first figure shows the MPPI algorithm sample uniformly in the configuration space, and the second figure shows  the MPPI-CBF algorithm sample in a restricted safe space (red: greater costs, blue: lower costs).}
\label{fig:single_plot}
\end{figure}



\subsection{Mutli-Obstacle Narrow Passage Problem}
We consider the same dynamical model as discussed in Section \ref{sec:uni-dyn} and further add multiple circular obstacles to simulate a narrow passage problem. The $n$ circular obstacles are located at $(O_{x_i},O_{y_i})$ with radius of $O_{r_i} (i=1,\dots,n)$, respectively. The distance between each circle is much smaller comparing to the radius. The safe set for the space is defined as: $C=\left \{(x,y)|(x-O_{x_i})^2 + (y-O_{y_i})^2 - O_{r_i} ^2 \geq 0, i = 1, \dots, n \right \}$.  We use the same cost function \eqref{eq:sim_cost_function} and similar control barrier functions $h_i(x,y)=(x-O_{x_i})^2 + (y-O_{y_i})^2 - O_{r_i} ^2$ in the single-obstacle avoidance scenario. The initial mean and variance of the noise distribution is: $\mu_0 = [0,0], \Sigma_0 =I$. We also use the same target position and desired velocity. For the MPPI algorithm, we use 40 time steps to guarantee the algorithm can reach the target position, and for the  MPPI-CBF algorithm, we use 20 time steps to save the computation time in the simulation.

The running trajectories solving the narrow passage problem using MPPI and MPPI-CBF algorithms are shown in Fig. \ref{fig:multi_traj}. The proposed MPPI-CBF algorithm admits safety guarantee (obstacle avoidance) and a shorter travel time to the target. However, with the same amount of samples, the MPPI algorithm cannot find the shortest path to the target because most sampled trajectories violate the safety constraints and get penalized. In the MPPI-CBF algorithm, however, almost all sample trajectories  lie in the safe set computed in the SDP optimization step, and thus the samples are more efficient, which explains why the MPPI-CBF algorithm needs fewer samples and time steps to solve the problem. As shown in Fig. \ref{fig:cost_multi}, with the same sample number, the MPPI-CBF algorithm has a more rapid cost reduction than the MPPI algorithm.
\begin{figure}
    \centering
    \includegraphics[width=6cm]{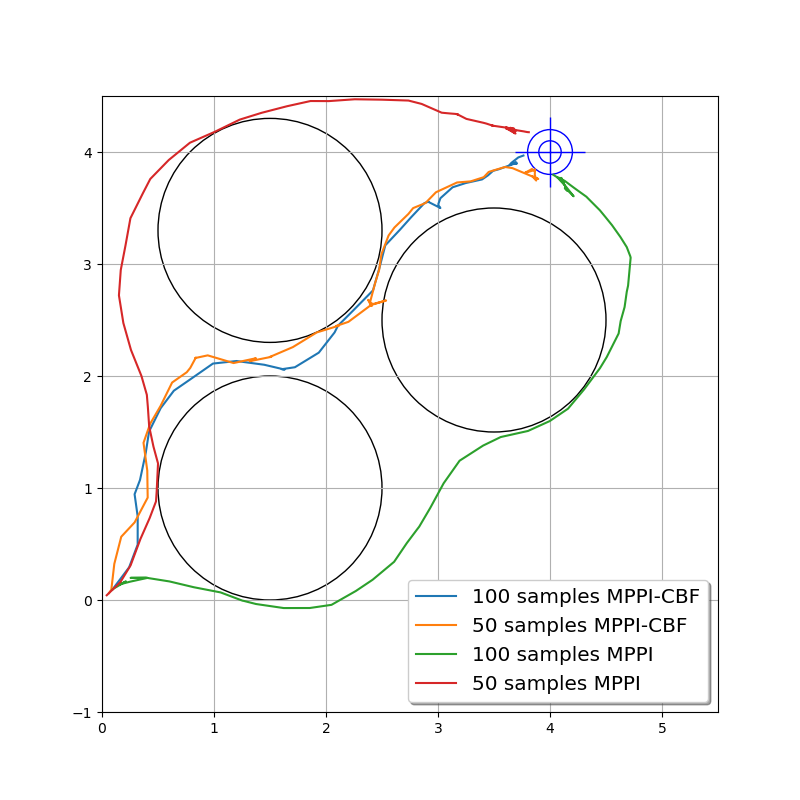}
    \caption{We use three circular obstacles to simulate a narrow passage problem, which is denoted by the black circles. The target is in blue. We plot the trajectories from the MPPI and the MPPI-CBF algorithms, with sample numbers of 50 and 100.}
    \label{fig:multi_traj}
\end{figure}
\begin{figure}
    \centering
    \includegraphics[width=6cm]{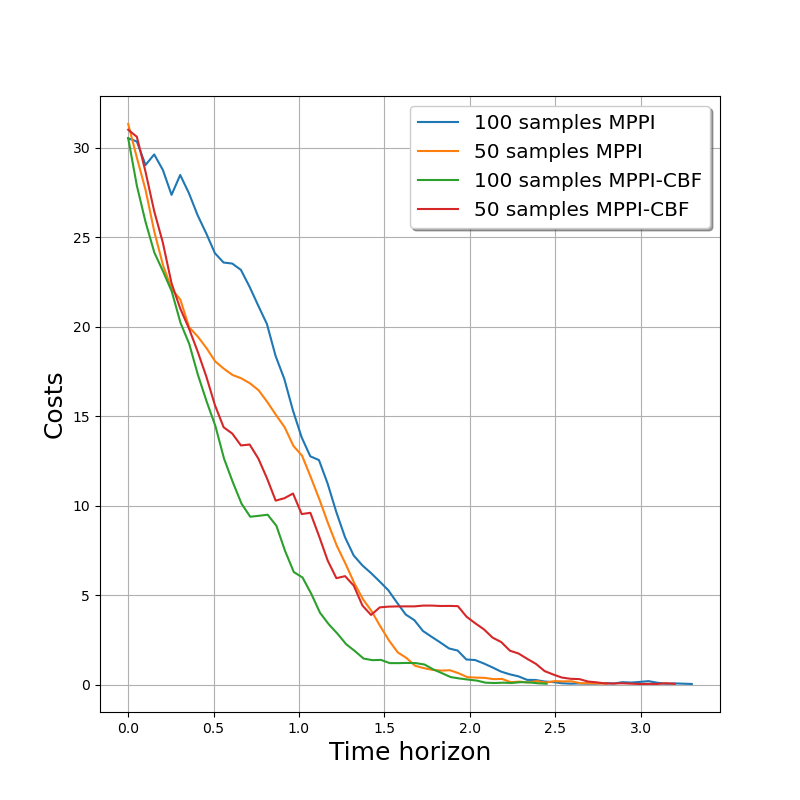}
    \caption{The cost of the MPPI algorithm and MPPI-CBF algorithm. In the multi-obstacle narrow passage simulation, the MPPI algorithm takes more time steps to solve the problem than the MPPI-CBF algorithm due to the higher sample efficiency of the latter approach.}
    \label{fig:cost_multi}
\end{figure}

We also plot the sample trajectories at time $t=3.5s$. The results are shown in Fig. \ref{fig:multi_MPPI_sample} and Fig. \ref{fig:multi_MPPI_CBF_sample}. The black circles represent the obstacles, and the sample trajectories are in different colors according to the costs. It is observed  that most sample trajectories of the MPPI algorithm violate the safe constraints. However, in the MPPI-CBF algorithm, almost all sample trajectories are in the safe region, which validates our claim on the sample efficiency of the proposed MPPI-CBF algorithm.

\begin{figure}[ht] 
\begin{subfigure}{.24\textwidth}
  \centering
  \includegraphics[width=1.2\linewidth]{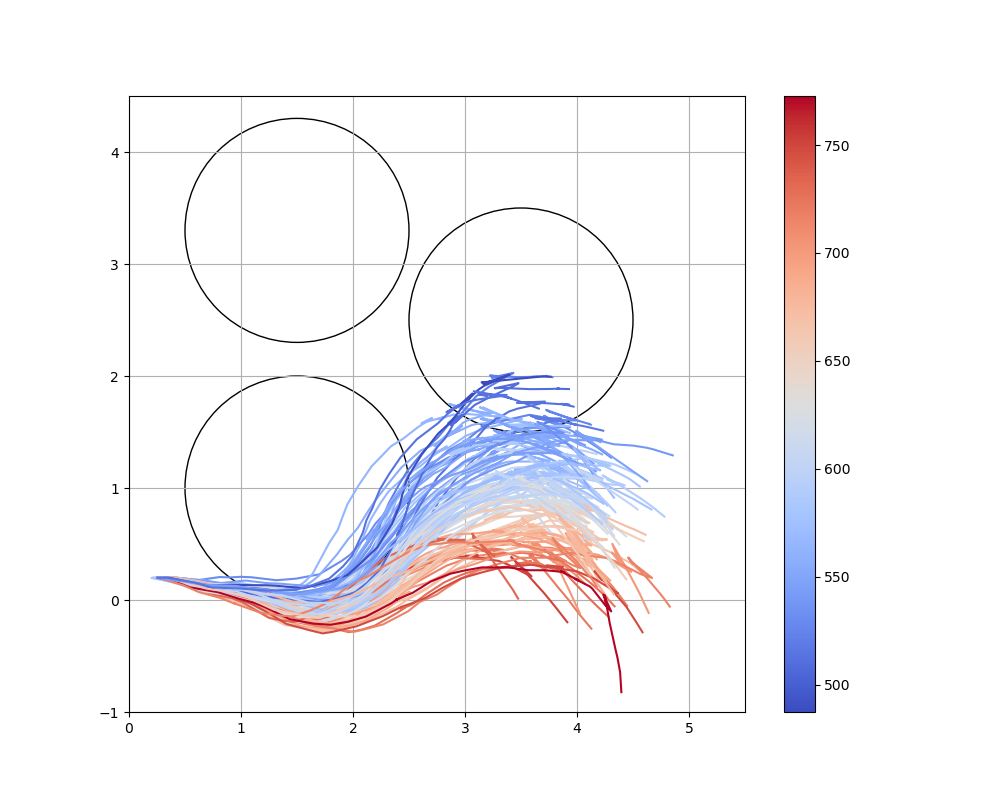}
  \caption{MPPI sample trajectories at time $t=0.35s$}
  \label{fig:multi_MPPI_sample}
\end{subfigure}
\begin{subfigure}{.24\textwidth}
  \centering
  \includegraphics[width=1.2\linewidth]{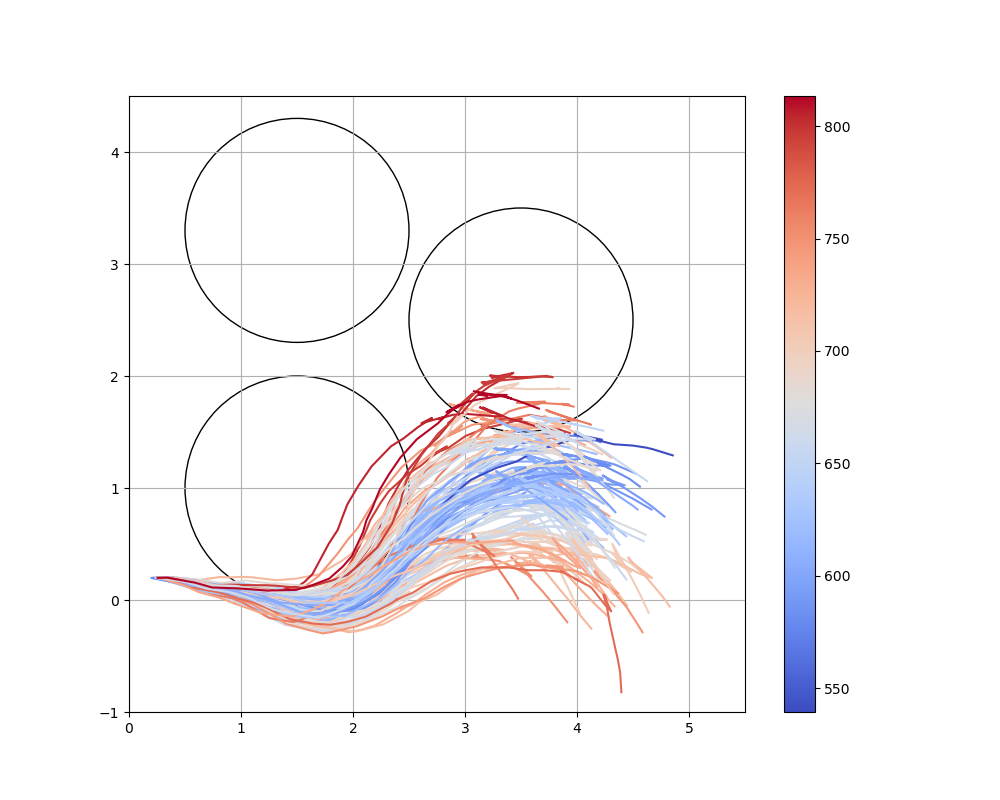}  
  \caption{MPPI-CBF sample trajectories at time $t=0.35s$}
  \label{fig:multi_MPPI_CBF_sample}
\end{subfigure}
\caption{The first figure shows that the MPPI algorithm samples uniformly in the configuration space, and the second figure shows that the MPPI-CBF algorithm samples from a restricted safe space (red: greater costs, blue: lower costs).}
\end{figure}

\section{Conclusion and Future Plan} \label{sec:con}
In this paper, we present an MPPI-CBF algorithm that uses the control barrier function to increase the safe efficiency and guarantee the safety of the sampling-based MPPI algorithm. We demonstrate that our algorithm can guarantee safety in cluttered environments, and our method can sample more efficiently. However, the algorithm requires solving SDP in each time step. The future work considers accelerating the algorithm by solving the SDP problem more efficiently and pursuing theoretical analysis on how the newly added state-dependent variance in the proposed algorithm deviates from the existing MPPI algorithm in terms of optimality.

\bibliographystyle{ieeetr}
\bibliography{citation}

\end{document}